\documentclass[smallextended]{svjour3}
\usepackage{llncsdoc}
\usepackage{amsmath}%,amsthm}
\usepackage{amssymb}
\usepackage{algorithm,algorithmic}
\usepackage{subfig}
\usepackage{xspace}
\usepackage{enumitem}
\usepackage{ifthen}
\usepackage{wrapfig}
\smartqed

\newboolean{arxiv}
\setboolean{arxiv}{true}

\usepackage{tikz}
\tikzstyle{vertex}=[auto=left,circle,draw=black,fill=black!25,minimum size=20pt,inner sep=0pt]

\newtheorem{myclaim}{Claim}
\def\cqedsymbol{\ifmmode$\lrcorner$\else{\unskip\nobreak\hfil
\penalty50\hskip1em\null\nobreak\hfil$\lrcorner$
\parfillskip=0pt\finalhyphendemerits=0\endgraf}\fi}

\newcommand{\goodchi}{\protect\raisebox{2pt}{$\chi$}}
\newcommand{\chilin}{\goodchi_{\text{lin}}}
\newcommand{\chicen}{\goodchi_{\text{cen}}}
\newcommand{\tmax}{t_{\text{max}}}

\DeclareMathOperator*{\argmax}{arg\,max}
\DeclareMathOperator*{\td}{td}
\newcommand{\Schaffer}{Sch\"{a}ffer\xspace}
\newcommand{\Schaffers}{{\Schaffer}'s\xspace}
\newcommand{\Nesetril}{Ne\v{s}et\v{r}il }

\title{Polynomial Treedepth Bounds in Linear Colorings}
\author{Jeremy Kun \and
	Michael P.~O'Brien \and
	Marcin Pilipczuk \and
	Blair D.~Sullivan}
\institute{J. Kun \at Google, Mountain View, CA, USA \\\email{jkun@google.com} \and
	M. P. O'Brien \at North Carolina State University, Raleigh, NC, USA \\\email{mpobrie3@ncsu.edu} \and
	M. Pilipczuk \at University of Warsaw, Warsaw, Poland \\\email{malcin@mimuw.edu.pl} \and
	B. D. Sullivan \at North Carolina State University, Raleigh, NC, USA \\\email{blair\_sullivan@ncsu.edu}}

\ifthenelse{\boolean{arxiv}}{\date{\today}}{}

\begin{document}
	\maketitle
	\begin{abstract}
	\emph{Low-treedepth} colorings are an important tool for algorithms that exploit structure in classes of bounded expansion; they guarantee subgraphs that use few colors have \emph{bounded treedepth}.
	These colorings have an implicit tradeoff between the total number of colors used and the treedepth bound, and prior empirical work suggests that the former dominates the run time of existing algorithms in practice.
	We introduce \emph{$p$-linear colorings} as an alternative to the commonly used $p$-centered colorings.
	They can be efficiently computed in bounded expansion classes and use at most as many colors as $p$-centered colorings.
	Although a set of $k<p$ colors from a $p$-centered coloring induces a subgraph of treedepth at most $k$, the same number of colors from a $p$-linear coloring may induce subgraphs of larger treedepth.
	We establish a polynomial upper bound on the treedepth in general graphs, and give tighter bounds in trees and interval graphs via constructive coloring algorithms.
	We also give a co-NP-completeness reduction for recognizing $p$-linear colorings and discuss ways to overcome this limitation in practice.
	\ifthenelse{\boolean{arxiv}}{This preprint extends results that appeared in~\cite{lc_wg}; for full proofs omitted from~\cite{lc_wg}, see previous versions of this preprint.}{Some of these results appeared previously in~\cite{lc_wg}.}

\end{abstract}

	\section{Introduction}
Algorithms for graph classes that exhibit \emph{bounded expansion} structure~\cite{boundedexpansion1,boundedexpansion2,boundedexpansion3,sparsity} offer a promising framework for efficiently solving many NP-hard problems on real-world networks.
The structural restrictions of bounded expansion, which allow for pockets of localized density in globally sparse graphs, are compatible with properties of many real-world networks such as clustering and heavy-tailed degree distributions.
Moreover, multiple random graph models designed to mimic these properties have been proven to asymptotically almost surely belong to classes of bounded expansion~\cite{demaine2014structural}.
From a theoretical perspective, graphs belonging to classes of bounded expansion can be characterized by \emph{low-treedepth colorings} of bounded size, i.e.~using only a small number of colors.
Roughly speaking, a low-treedepth coloring is one in which the subgraphs induced on each small set of colors have small \emph{treedepth}, a structural property stronger than treewidth.
This definition naturally implies an algorithmic pipeline~\cite{boundedexpansion2,demaine2014structural,dvorak2013testing} for classes of bounded expansion involving four stages:  computing a low-treedepth coloring, using the coloring to decompose the graph into subgraphs of small treedepth, solving the problem efficiently on each such subgraph, and combining the subsolutions to construct a global solution.
The complexities of algorithms using this paradigm often are of the form $O({k \choose p} 2^{d\log d}\cdot n^c)$ where $k$ is the coloring size and $d$ is the treedepth of the subgraphs.

A recent implementation~\cite{concuss} and experimental evaluation~\cite{concuss_benchmarking} of this pipeline has identified that the coloring size has a much larger effect on the run time than the treedepth in practice.
Although graphs in classes of bounded expansion are guaranteed to admit colorings of constant size with respect to the number of vertices, the only known polynomial-time algorithms for computing these colorings are approximations~\cite{sparsity}.
Consequently it is unclear to what extent our current coloring algorithms can be altered to reduce the coloring size.
A more viable approach to improving the performance of the algorithmic pipeline without significant high-level changes would be to develop a new type of low-treedepth coloring that uses fewer colors but potentially has weaker guarantees about the treedepth of the subgraphs.

The traditional low-treedepth colorings for classes of bounded expansion are known as \emph{$p$-centered colorings}.
This name stems from the property that on any subgraph $H$, a $p$-centered coloring either uses at least $p$ colors or is a \emph{centered coloring}, which restricts the multiplicity of colors in induced subgraphs.
In this paper we introduce an alternative that closely mirrors this paradigm but only extends the color multiplicity guarantees to path subgraphs.
For this reason we refer to them as \emph{$p$-linear colorings} and \emph{linear colorings}.
We identify that $p$-linear colorings share three important properties with $p$-centered colorings that allow them to be used in the bounded expansion algorithmic pipeline.
\begin{enumerate}
	\item The minimum coloring size is constant in graphs of bounded expansion.
	\item A coloring of bounded size can be computed in polynomial time.
	\item Small sets of colors induce graphs of small treedepth.
\end{enumerate}
The third of these properties is of particular interest, since understanding the tradeoffs between coloring size and treedepth in switching between $p$-centered and $p$-linear colorings fundamentally depends on bounding the maximum treedepth of a graph that admits a linear coloring with $k$ colors.
Equivalently, we frame this problem as determining the gap between the minimum number of colors needed for a linear versus a centered coloring in any given graph.
Using a grid minors approach, we prove that the minimum size of a centered coloring is polynomially bounded in the minimum size of a centered coloring.
Because the ``heavy machinery'' of this approach likely does not give a tight bound, we give stronger upper bounds on the gap in trees and interval graphs and a matching lower bound for binary trees.
Surprisingly, we also prove that some $p$-linear colorings cannot be verified in polynomial time unless $\text{P} = \text{co-NP}$ and discuss the practical implications of these findings.
Some results in this paper appeared previously in WG 2018~\cite{lc_wg}.
This version adds a polynomial treedepth upper bound for general graphs, as well as tighter lower and upper bounds for trees.

	\section{Definitions and Background}
In this section we detail the background and terminology necessary to understand $p$-linear colorings.

\subsection{Graph Terminology}
	We denote the vertices and edges of a graph $G$ as $V(G)$ and $E(G)$, respectively, and assume all graphs are simple and undirected except where specifically noted otherwise.
	The \emph{open neighborhood} of a vertex $v$, denoted $N(v)$, is the set of vertices $u$ such that $uv\in E(G)$, while the \emph{closed neighborhood}, $N[v]$ is defined as $N(v)\cup \{v\}$.
	Vertex $a$ is an \emph{apex} with respect to a subgraph $H$ if $V(H) \subseteq N(a)$.

	We say $P$ is a \emph{$v_1v_\ell$-path} if $V(P) = \{v_1,\dots, v_\ell\}$ for distinct $v_1,\dots, v_\ell$ and $E(P) = \{v_iv_{i+1} : 1\leq i \leq \ell-1\}$; we will notate this as $P=v_1,\dots, v_\ell$.
	Given disjoint paths $P = v_1,\dots, v_\ell$ and $Q = u_1,\dots, u_{\ell'}$, the path $P\cdot Q = v_1,\dots, v_\ell, u_1,\dots, u_{\ell'}$ is the \emph{concatenation} of $P$ and $Q$ if $v_\ell$ and $u_1$ are adjacent.
	A path is \emph{Hamiltonian} with respect to subgraph $H$ if $V(P) = V(H)$.

	In a rooted tree $T$, we let $T_v$ be the subtree of $T$ rooted at $v$ and the \emph{leaf paths} of $T_v$ be the set of paths from a leaf of $T_v$ to $v$.
	We label the levels of $T$ from bottom to top starting from 1; that is, if $D$ is the maximum distance from a leaf to the root then the root is the only vertex in level $D+1$ and level $i$ consists of all vertices whose parents are in level $i+1$.
	Vertices $u$ and $v$ are \emph{unrelated} in $T$ if $u$ is neither an ancestor nor a descendant of $v$.

	A \emph{coloring} $\phi$ of a graph $G$ is a mapping of the vertices of $G$ to \emph{colors} $1,\dots, k$ and has \emph{size} $|\phi|=k$.
	A coloring is \emph{proper} if no pair of adjacent vertices have the same color.
	For any subgraph $H$ and color $c$, if there is exactly one vertex $v\in H$ such that $\phi(v) = c$ we say $c$ \emph{appears uniquely} in $H$ and $v$ is a \emph{center} of $H$.
	A subgraph with no unique color is said to be \emph{non-centered}.

	We use the notation $X = Y_1\uplus\dots \uplus Y_\ell$ to denote that $Y_1,\dots, Y_\ell$ form a \emph{partition} of $X$; that is, $X = Y_1\cup\dots \cup Y_\ell$ and the sets $Y_1,\dots Y_\ell$ are pairwise disjoint.

\subsection{$p$-Centered Colorings and Bounded Expansion}
	\begin{definition}
		A \emph{$p$-centered coloring} $\phi$ of graph $G$ is a coloring such that for every connected subgraph $H$, $H$ has a center or $\phi|_{H}$ uses at least $p$ colors.
	\end{definition}
	\Nesetril and Ossana de Mendez established that bounding the minimum size of a $p$-centered coloring is a necessary and sufficient condition for a graph class to have bounded expansion.
	\begin{proposition}[\cite{boundedexpansion1}]
		A class of graphs $\mathcal C$ has bounded expansion iff there exists a function $f$ such that for all $G\in \mathcal C$ and all $p\geq 1$, $G$ admits a $p$-centered coloring with $f(p)$ colors.
	\end{proposition}
	There are varying methods to compute $p$-centered colorings, such as transitive-fraternal augmentations~\cite{boundedexpansion1,grohe2017deciding} and generalized coloring numbers~\cite{zhu2009colouring}, we focus here on \emph{distance-truncated transitive-fraternal augmentations} (DTFAs)~\cite{FelixThesis}, which iteratively augment the graph with additional edges to impose constraints on proper colorings.
	This linear time algorithm guarantees that after $(2\log p)^p$ DTFA iterations, any proper coloring of the augmented graph is a $p$-centered coloring whose size is  bounded in classes of bounded expansion.

\subsection{Centered Colorings and Treedepth}
	Note that if $\phi$ is a $p$-centered coloring of $G$ and $H$ is a subgraph of $G$ whose vertices use at most $p-1$ colors in $\phi$, $H$ must have a center.
	This relates $p$-centered colorings to a more restricted class of graphs defined by \emph{centered colorings}.
	\begin{definition}\label{def:centered_coloring}
		A \emph{centered coloring} $\phi$ of graph $G$ is a coloring such that every connected subgraph has a center.
		The minimum size of a centered coloring of $G$ is denoted $\chicen(G)$.
	\end{definition}
	Note that a centered coloring is also proper, or else there would be a connected subgraph of size two with no center.
	Observe that if $X$ is the set of all centers of $G$, then $G\backslash X$ must either be empty or disconnected.
	This implies that if $|G|\gg \chicen(G)$, then $G$ breaks into many components after only a few vertex deletions.
	This property is captured by \emph{treedepth decompositions}.
	\begin{definition}
		A \emph{treedepth decomposition} $\mathcal T$ of graph $G$ is a rooted forest with the same vertex set as $G$ such that $uv\in E(G)$ implies $u$ is an ancestor of $v$ in $\mathcal T$ or vice versa.
		The \emph{depth} of $\mathcal T$ is the length of the longest path from a leaf of $\mathcal T$ to the root of its component.
		The \emph{treedepth} of $G$, $\td(G)$, is the minimum depth of a treedepth decomposition of $G$.
	\end{definition}

	Given a centered coloring of size $k$, we can generate a treedepth decomposition of depth at most $k$ by choosing any center $v$ to be the root and setting the children of $v$ to be the roots of the treedepth decompositions of the components of $G\backslash\{v\}$.
	Likewise, given a treedepth decomposition of depth $k$, we can generate a centered coloring using $k$ colors by bijectively assigning the colors to levels of the tree and coloring vertices according to their level.
	We refer to the colorings and decompositions resulting from these procedures as \emph{canonical}; together they imply that the treedepth and centered coloring numbers are equal for all graphs.

\section{$p$-Linear and Linear Colorings}
	We introduce \emph{$p$-linear colorings} as an alternative to $p$-centered colorings.
	\begin{definition}\label{def:p_linear}
		A \emph{$p$-linear coloring} is a coloring $\psi$ of a graph $G$ such that for every path\footnote{This includes non-induced paths.} $P$, either $P$ has a center or $\psi|_P$ uses at least $p$ colors.
	\end{definition}
	It is proven in~\cite{FelixThesis} that after performing $2^p$ DTFA iterations, any proper coloring of the augmented graph is a $p$-linear coloring.
	This implies that $p$-linear colorings indeed have constant size in bounded expansion classes and can be constructed in polynomial time (like $p$-centered colorings).

	In the interest of maintaining consistency with prior terminology, we define \emph{linear colorings} analogously to centered colorings.
	\begin{definition}\label{def:linear_coloring}
		A \emph{linear coloring} is a coloring $\psi$ of a graph $G$ such that every path has a center.
		The \emph{linear coloring number} is the minimum number of colors needed for a linear coloring and is denoted $\chilin(G)$.
	\end{definition}
	Note that linear colorings must also be proper.
	A simple recursive argument shows that every path of length $d$ requires at least $\log_2(d+1)$ colors in a linear coloring; thus a graph of linear coloring number $k$ has no path of length $2^k$.
	Because every depth-first search tree is a treedepth decomposition, $\td(G) \leq 2^{\chilin(G)}$, proving that small numbers of colors in $p$-linear colorings induce graphs of bounded treedepth\footnote{This tightens a bound in~\cite{FelixThesis} from double to single exponential.}.

	Our study of the divergence between linear and centered coloring numbers will naturally focus on linear colorings that are not also centered colorings.
	We say $\psi$ is a \emph{non-centered linear coloring} (NCLC) of graph $G$ if $G$ contains a connected induced subgraph with no center.
	For NCLC $\psi$, we say a connected induced subgraph $H$ is a \emph{witness} to $\psi$ if $H$ is non-centered but every proper connected subgraph of $H$ has a center.
	For the sake of completeness, we prove in Lemma~\ref{lem:lin_cen_eq} that many simple graph classes do not admit NCLCs.
	\begin{lemma}\label{lem:lin_cen_eq}
		If $G$ is a cograph, has maximum degree 2, or has independence number 2, any linear coloring of $G$ is also a centered coloring.
	\end{lemma}
	\begin{proof}
		We analyze each graph class separately below.\\
		\noindent\textbf{Maximum degree 2:}
			Let $G$ be a graph of maximum degree 2.
			Each connected induced subgraph of $G$ is either a path or a cycle, both of which have a Hamiltonian path.
			Thus every connected subgraph has a center, making any linear coloring centered.

			\noindent\textbf{Cographs:}
			Let $\psi$ be an NCLC of cograph $G$ and $H$ be a witness to $\psi$.
			If $\psi|_H$ only contains one color, $H$ is an isolated vertex and the coloring is centered.
			Thus, we may assume $\psi|_H$ has at least two colors.
			Because $H$ is a cograph, we can partition its vertices into nonempty sets $X, Y$ such that $xy$ is an edge in $H$ for all $x\in X$ and $y\in Y$.
			But since $\psi$ is proper, every pair of vertices with the same color must lie in the same set $X$ or $Y$.
			Since every color in $\psi|_H$ appears at least twice, there are vertices $\{v,v'\}\in X$ and $\{u,u'\}\in Y$ such that $\psi(v) = \psi(v')$ and $\psi(u) = \psi(u')$ but $\psi(v)\neq \psi(u)$.
			But then $v,u,v'u'$ form a path with no center and thus $\psi$ is not a linear coloring.

			\noindent\textbf{Independence number 2:}
			Since independence number is hereditary, it is sufficient to show every connected graph of independence number 2 has a Hamiltonian path.
			We prove this by induction on the number of vertices, observing that an isolated vertex has a trivial Hamiltonian path.
			Let $G$ be a graph of independence number 2 and $v\in G$ a vertex such that $G\setminus \{v\}$ is connected, e.g., $v$ is a leaf in a minimum spanning tree of $G$.
			If $G\setminus \{v\}$ has a Hamiltonian cycle, then $G$ must have a Hamiltonian path.
			Otherwise, by the inductive hypothesis $G\setminus \{v\}$ has a Hamiltonian path whose endpoints are some non-adjacent pair of vertices $u,w$.
			Either $v$ is adjacent to one of $u, w$, in which case $G$ has a Hamiltonian path, or $\{u, w, v\}$ form an independent set of size 3.
	\hfill\qed\end{proof}
	The classes described in Lemma~\ref{lem:lin_cen_eq} are maximal in the sense that there are graphs with independence number 3 (graph $R_3$ described in Lemma~\ref{lem:recursive_clique}) and binary trees (Lemma~\ref{lem:binary_tree_lc}) that admit NCLCs.

	\section{Treedepth Lower Bounds}
To understand the tradeoff between the number of colors and treedepth of small color sets when using $p$-linear colorings in lieu of $p$-centered colorings, it is important to know the maximum treedepth of a graph of fixed linear coloring number $k$, $\tmax(k)$.
In Lemmas~\ref{lem:recursive_clique} and~\ref{lem:binary_tree_lc}, we prove lower bounds on $\tmax(k)$ through explicit constructions of graph families.
In order to show that these graphs have large treedepth, we first establish assumptions about the structure of treedepth decompositions that can be made without loss of generality.
\begin{lemma}\label{lem:apex_order}
	Let $G$ be a graph and $S\subset V(G)$ such that $G[S]$ is connected and with respect to some component $C\in G\backslash S$, every vertex in $S$ is an apex of $C$.
	Then for any treedepth decomposition $\mathcal T$ of $G$ with depth $k$, we can construct a treedepth decomposition $\mathcal T'$ such that:
	\begin{enumerate}
		\item $\operatorname{depth}(\mathcal T') \leq k$
		\item Each vertex in $S$ is an ancestor of every vertex in $C$ in $\mathcal T'$
		\item For each pair of vertices $\{u,w\}\subseteq V(C)$ or $\{u,w\}\subseteq V(G\setminus C)$, $u$ is an ancestor of $w$ in $\mathcal T'$ iff it is an ancestor of $w$ in $\mathcal T$.
	\end{enumerate}
\end{lemma}
\begin{proof}
	%Color the vertices of $G$ according to the levels of $T$ i.e. the root of $T$ gets color 1, its children get color 2, etc.
	%By definition this yields a centered coloring $\psi$.
	Let $\phi$ be a canonical centered coloring of $G$ with respect to $\mathcal T$.
	Let $\mathcal T'$ be a canonical treedepth decomposition with respect to $\phi$; if there are multiple vertices of unique color, prioritize removing those outside $C$ before members of $C$, and then small colors over large colors, i.e., remove color 2 before color 5.
	Since $\mathcal T'$ is derived from a centered coloring with $k$ colors, its depth is at most $k$, satisfying condition 1.

		Condition 2 is satisfied as long each member of $S$ is removed in the construction of $\mathcal T'$ before any member of $C$.
	Note that since $S$ contains apex vertices with respect to $C$ and every vertex $v\in V(C)$ satisfies $N[v]\subseteq V(C)\cup V(S)$, the removal of any vertex from $C$ cannot disconnect a previously connected component if $S$ has not been removed.
	Thus at any point in the algorithm before the removal of $S$ if a vertex in $C$ has a unique color in its remaining component $H$, there must be another vertex in $H\backslash C$ of unique color as well.
	Consequently, we will never be forced to remove any vertex of $C$ before $S$.

		To prove condition 3 is satisfied, observe that $u$ is an ancestor of $w$ in $\mathcal T'$ iff there is a connected subgraph $H$ containing $u$ and $w$ and no vertex with color smaller than $\psi(u)$.
	As stated previously, $G\backslash C$ is a connected subgraph, which means that there is a subgraph witnessing this ancestor-descendant relationship between $u$ and $w$ such that $H\cap C = \emptyset$ if $u\notin C$ and $H\cap (G\backslash C) = \emptyset$ if $u\in C$.
	Thus the relationships in $\mathcal T$ are preserved in $\mathcal T'$.
\hfill\qed\end{proof}

Using Lemma~\ref{lem:apex_order}, we now show that $\tmax(k)\geq 2k$.
\begin{lemma}\label{lem:recursive_clique}
	There exists an infinite sequence of graphs $R_1,R_2,\dots$ such that
	$$
		\lim_{i\to\infty} \frac{\chicen(R_i)}{\chilin(R_i)} = 2.
	$$
\end{lemma}
\begin{proof}
	Define $R_{i}$ recursively such that $R_0$ is the empty graph and $R_{i}$ is a complete graph on vertices $v_1,\dots, v_{i}$ along with $i$ copies of $R_p$ for $p=\lfloor \frac{i-1}{2}\rfloor$, call them $H_1,\dots, H_i$, such that $v_j$ is an apex with respect to $H_j$ (Figure 1).
	We prove that $\chilin(R_i) = i$ and $\lim_{i\to\infty} \chicen(R_i) = 2i$.

		With respect to the linear coloring number, note that $\chilin(R_i)\geq i$ since the clique of size $i$ requires $i$ colors by Lemma~\ref{lem:lin_cen_eq}.
	We prove the upper bound $\chilin(R_i)\leq i$ by induction on $i$.
	The case of $i=1$ is trivial; assume it is true for $1,\dots, i-1$.
	From the inductive hypothesis, we can assume each $H_j$ only requires $p$ colors for a linear coloring.
	Consider the coloring $\psi$ of $R_i$ such that $\psi(v_j) = j$ and $\psi|_{H_j}$ is a linear coloring of $H_j$ using colors $\{1+(j+1)\bmod i, 1+(j+2)\bmod i,\dots, 1+(j+p)\bmod i\}$.
	If $\psi$ is not a linear coloring, there is some path $Q$ without a center.
	Since $\psi(v_j)\notin \psi|_{H_j}$, $Q$ must contain vertices from at least two $H_j$s; each $v_j$ is a cut vertex, so $Q$ cannot contain vertices from more than two $H_j$s.
	However, $\psi^{-1}(1)\subseteq \{v_{1}\} \cup V(H_{2}) \cup \dots \cup V(H_{p+1})$, but $\{2,\dots, p+1\}\notin \psi|_{H_1}$, which means $Q\cap H_1 = \emptyset$.
	Based on the symmetry of $\psi$ we can apply the same argument to the remaining colors, which means that no such non-centered path $Q$ exists and $\psi$ is indeed a linear coloring of size $i$.

		With respect to the centered coloring number, by Lemma~\ref{lem:apex_order} there is an minimum-depth treedepth decomposition in which $v_j$ is an ancestor of $H_j$.
	This implies there is a $j$ such that no vertex in $H_j$ shares a color in the canonical coloring with any of the vertices in the clique.
	Thus $\chicen(R_i) = i + \chicen(R_p)$; in the limit this recursion approaches $2i$.
\hfill\qed\end{proof}

The graphs in Lemma~\ref{lem:recursive_clique} contain large cliques.
We now show that this is not a necessary condition for the linear and centered coloring numbers to diverge.
\begin{lemma}\label{lem:binary_tree_lc}
	Let $B_\ell$ be the complete binary tree with $\ell$ levels.
	Then
	$$
		\lim_{\ell\to\infty} \frac{\chicen(B_\ell)}{\chilin(B_\ell)} \geq \log_2 3.
	$$
\end{lemma}

\begin{proof}
Fix an integer $a \geq 1$ and let $b$ be the smallest integer such that
\begin{equation}\label{eq:bt:1}
2^a < 3^b.
\end{equation}
Our proof proceeds by first constructing a \emph{coloring pattern} $\Psi_a$ of $B_a$ and then using $\Psi_a$ to create a linear coloring for an arbitrarily large complete binary tree.
Some vertices of $B_a$ will be left uncolored (we will call them \emph{local}), while some vertices will be colored with one of
the $b$ colors $[b]$ (we will call these colors \emph{global}).
Let $C_1,C_2,\ldots,C_{2^b}$ be the sequence of all subsets of $[b]$ in order of nonincreasing size (in particular, $C_1 = [b]$ and $C_{2^b} = \emptyset$)
  and let $\ell$ be such that $\sum_{i=1}^\ell 2^{|C_i|-1} = 2^{a-1}$. Note that such an index $\ell$ exists due to Equation~\eqref{eq:bt:1}:
 $$\sum_{i=1}^{2^b} 2^{|C_i|-1} = \frac{1}{2} \cdot 3^b > 2^{a-1}$$
 and the fact that the sets $C_i$ are ordered in the nonincreasing order of their sizes. Furthermore, we have $\ell < 2^b$.

Let $v_1,v_2,\ldots,v_{2^{a-1}}$ be an ordering of the leaves of $B_a$ corresponding to an in-order traversal.
Consider an index $1 \leq i \leq \ell$. By construction, there exists a vertex $v_i\in B_a$ at level $|C_i|$ that is the root of a subtree $T_{v_i}$ whose leaves are exactly $v_j$ for $\sum_{i'=1}^{i-1} 2^{|C_{i'}|-1} < j \leq \sum_{i'=1}^i 2^{|C_{i'}|-1}$.
We color the vertices of $T_{v_i}$ level by level with (global) colors of $C_i$; that is, we order the colors of $C_i$ arbitrarily
and color level $k$ of $T_{v_i}$ with the $k$-th color of $C_i$ for every $k \in [|C_i|]$.
All remaining vertices of $B_a$ (that is, those that lie in none of the subtrees $T_{v_i}$ for $1 \leq i \leq \ell$) remain local.

The following claim summarizes the properties of the above coloring.
\begin{myclaim}\label{cl:bt:pattern}
For every path $P$ in $B_a$
that either
\begin{itemize}
\item has both endpoints in a leaf or the root of the tree $B_a$, or
\item does not contain a local vertex,
\end{itemize}
there exists a global color $c \in [b]$ such $c$ appears uniquely on $P$.
\end{myclaim}
\begin{proof}
If a path $P$ does not contain a local vertex, then it is contained in a single tree $T_{v_i}$.
For such a path, the unique vertex on $P$ of maximum level
  is colored with a global color that appears uniquely on $P$.
Similarly, if $P$ is a leaf path in $B_a$, then any globally colored vertex
of the tree $T_{v_i}$ containing the leaf endpoint of $P$ satisfies the desired property.
Otherwise, a path $P$ that has both endpoints in leaves of $B_a$ but contains a local vertex needs
to start in a leaf of one subtree $T_{v_i}$ and end in a leaf of a different subtree $T_{v_i'}$.
Then, observe that any (global) color of $C_i \triangle C_{i'}$ appears exactly once on $P$.
\cqedsymbol\end{proof}

Let $p < 2^a$ be the number of local vertices in the pattern $\Psi_a$.
For an even integer $d \geq 2$, consider a coloring $\psi$ of $B_{ad}$ defined as follows.
Fix a palette $[db]$ of \emph{global} colors and a palette $[2p]$ of \emph{local} colors.
For every $1 \leq i \leq d$, the $i$-th \emph{stripe} consists
of $a$ levels $(i-1)a+1,\ldots,ia$.
In $B_{ad}$, such a stripe consists of $2^{(d-i)a}$ copies of $B_a$.
Color every such copy using the pattern $\Psi_a$
with global colors $(i-1)b+1, \ldots, ib$
as the $b$ global colors of $\Psi_a$ and color each local vertex with a different local
color from the set $\{1,2,\ldots,p\}$ if $i$ is odd and from the set $\{p+1,p+2,\ldots,2p\}$
if $i$ is even.

We claim that the above is a linear coloring of $B_{ad}$ with $db+2p < db+2^{a+1}$ colors.
Consider a path $P$ in $B_{ad}$ and let $i$ be the index of the highest stripe
intersected by $P$.
By the choice of $i$, $P$ intersects exactly one of the copies of $B_a$ in the $i$-th stripe.
If $P$ contains a leaf-to-leaf path in this copy, then Claim~\ref{cl:bt:pattern} asserts
that $P$ contains a center in this copy (recall that every stripe uses a different
    set of $b$ global colors).
Otherwise, $P$ intersects at most one copy of $B_a$ in every stripe.
If $P$ intersects at least three stripes, then $P$ contains a root-to-leaf path
in the single copy of $B_a$ intersected by $P$ at stripe $(i-1)$, and we are
again done by Claim~\ref{cl:bt:pattern}.
Similarly, Claim~\ref{cl:bt:pattern} finishes the proof if $P$ does not contain a local
vertex at the $i$-th stripe.
Finally, in the remaining case $P$ intersects at most two stripes (the $i$-th one
and possibly the $(i-1)$-th one) and contains a local vertex in the $i$-th stripe.
Since we used different set of local colors for odd and even stripes, any such local
vertex in $i$-th stripe is a center of $P$.

Consequently, we have exhibited a linear coloring of $B_{ad}$ with less than $db+2^{a+1}$
colors, where $b$ is defined as in Equation~\eqref{eq:bt:1}. If we let $d$ go to $\infty$, then
the ratio $(ad)/(db+2^{a+1})$ approaches $a/b$. This ratio, in turn, approaches
$\log_2(3)$ as $a \to \infty$ due to the choice of $b$ at Equation~\eqref{eq:bt:1}.
This finishes the proof of the lemma.
\qed\end{proof}

\begin{figure}
	\centering
	% \subfloat[$R_6$ (Lemma~\ref{lem:recursive_clique}).]{
		\resizebox{0.5\linewidth}{!}{\begin{tikzpicture}
	% Clique of 5
	\node[vertex,fill=blue]   (c01) at (120:2) {};
	\node[vertex,fill=red]    (c02) at ( 60:2) {};
	\node[vertex,fill=black]  (c03) at (  0:2) {};
	\node[vertex,fill=green]  (c04) at (300:2) {};
	\node[vertex,fill=yellow] (c05) at (240:2) {};
	\node[vertex,fill=white]  (c06) at (180:2) {};

	% Clique of 2 from blue
	\node[vertex,{fill=red}]    (c11) at (130:4) {};
	\node[vertex,{fill=black}]  (c12) at (110:4) {};

	% Clique of 2 from red
	\node[vertex,{fill=black}]  (c21) at (70:4) {};
	\node[vertex,{fill=green}]  (c22) at (50:4) {};

	% Clique of 2 from black
	\node[vertex,{fill=green}]  (c31) at ( 10:4) {};
	\node[vertex,{fill=yellow}] (c32) at (350:4) {};

	% Clique of 2 from green
	\node[vertex,{fill=yellow}] (c41) at (310:4) {};
	\node[vertex,{fill=white}]  (c42) at (290:4) {};

	% Clique of 2 from yellow
	\node[vertex,{fill=white}]  (c51) at (250:4) {};
	\node[vertex,{fill=blue}]   (c52) at (230:4) {};

	% Clique of 2 from white
	\node[vertex,{fill=blue}]   (c61) at (190:4) {};
	\node[vertex,{fill=red}]    (c62) at (170:4) {};

	\foreach \from/\to in	{c01/c02,
							 c02/c03,
							 c03/c04,
							 c04/c05,
							 c05/c06,
							 c11/c12,
							 c21/c22,
							 c31/c32,
							 c41/c42,
							 c51/c52,
							 c61/c62}
	\draw[thick,black!60] (\from) -- (\to);

	\foreach \from/\to in	{c01/c11,
							 c01/c12,
							 c03/c31,
							 c03/c32,
							 c01/c03,
							 c01/c05,
							 c02/c04,
							 c02/c06,
							 c03/c05,
							 c04/c06}
	\path[thick,black!60] (\from) edge (\to);

	\foreach \from/\to in	{c01/c04,
							 c02/c05,
							 c03/c06,
							 c01/c06,
							 c02/c21,
							 c02/c22,
							 c04/c41,
							 c04/c42}
	\path[thick,black!60] (\from) edge (\to);

	\foreach \from/\to in	{c05/c51,c05/c52,c06/c61,c06/c62}
	\draw[thick,black!60] (\from) -- (\to);

\end{tikzpicture}}
		%  }
	% \subfloat[$B_3$ (Lemma~\ref{lem:binary_tree_lc}).]{
	% 	\resizebox{0.5\linewidth}{!}{\input{figs/binary_tree_lc}}}
	\caption{Linear colorings of graph $R_6$ in Lemma~\ref{lem:recursive_clique}}\label{fig:extremal_lc}
\end{figure}
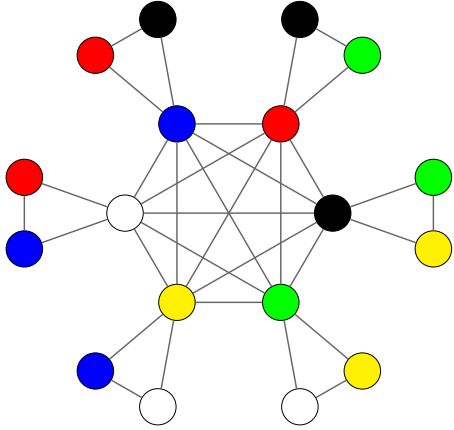
In Section~\ref{sec:trees} we show that the bound in Lemma~\ref{lem:binary_tree_lc} is tight for binary trees (Theorem~\ref{thm:tree_bound}).
We conjecture that the construction in Lemma~\ref{lem:recursive_clique} is also tight for general graphs.
\begin{conjecture}\label{conj:two_chilin}
	For any graph $G$, $\chicen(G)\leq 2\chilin(G)$.
\end{conjecture}
While the exclusion of a path of length $2^k$ indicates $\tmax(k) \leq 2^k$, this nonetheless leaves a large gap between the upper and lower bounds on $\tmax(k)$.
To move towards a proof of Conjecture~\ref{conj:two_chilin}, we establish a polynomial upper bound on $\tmax(k)$ in general graphs in the next section (Theorem~\ref{thm:grid_bound}).
Because this proof uses ``heavy machinery'', we consider two restricted graph classes\textemdash namely, trees and interval graphs\textemdash in Sections~\ref{sec:trees} and~\ref{sec:interval} and give tighter upper bounds on $\tmax(k)$ for graphs in these classes.

	\section{Treedepth Upper Bounds on General Graphs}
This section is devoted to proving a polynomial upper bound on $\tmax(k)$.
\begin{theorem}\label{thm:grid_bound}
	There exists a polynomial $p$ such that every graph $G$ satsifies
  $\chicen \leq \chilin^{190} p(\log \chilin)$.
\end{theorem}

Our starting point is the following theorem of Kawarabayashi and Rossman~\cite{KawarabayashiR18}:
\begin{theorem}[\cite{KawarabayashiR18}]
There is an absolute constant $C$ such that every graph $G$
of treedepth at least $C k^5 \log^2 k$ satisfies at least one of the following:
\begin{enumerate}
\item the treewidth of $G$ is at least $k$;
\item $G$ contains a complete binary tree of height $k$ as a minor;
\item $G$ contains a path on $2^k$ vertices.
\end{enumerate}
\end{theorem}
Assume that the treedepth of $G$ is at least $C k^5 \log^2 k$.
If $G$ contains a path on $2^k$ vertices (condition 3), then clearly $\chilin(G) \geq k$.
If $G$ contains a complete binary tree of height $k$ as a minor (condition 2), then $G$
also contains a subdivision of a complete binary tree of height $k$ as a subgraph.
Since $\chilin(H) \leq \chilin(G)$ for any subgraph $H$ of $G$,
Theorem~\ref{thm:tree_bound} asserts that $\chilin(G) \geq k / \log_2(3)$.
Thus, in the proof of Theorem~\ref{thm:grid_bound}, we are left with the
case when $G$ has large treewidth.

Here, we use the celebrated grid minor theorem, with the best known bound due to Chuzhoy~\cite{Chuzhoy16}.
\begin{theorem}[\cite{Chuzhoy16}]
	There is a polynomial $p'$ such that every graph $G$ with treewidth
	at least $k^{19} p'(\log k)$ contains a $k \times k$ grid as a minor.
\end{theorem}
We slightly relax the notion of a $k \times k$ grid minor to a \emph{$k$-pseudogrid}, defined as follows.
\begin{definition}
	A graph $G$ \emph{contains a $k$-pseudogrid} if there exist two sequences of vertex-disjoint
	paths in $G$, $\mathcal{P} = (P_1,P_2,\ldots,P_k)$
	and $\mathcal{Q} = (Q_1,Q_2,\ldots,Q_k)$ such that
	\begin{itemize}
		\item for every $i \in [k]$, the path $P_i$ is a concatenation of paths $P_{i,0}$, $P_{i,1}^Q$, $P_{i,1}$, $P_{i,2}^Q$, $P_{i,2}$, $\ldots$, $P_{i,k}^Q$, $P_{i,k}$ in this order such that each path $P_{i,j}^Q$ for $j \in [k]$ is a subpath of $Q_j$ (possibly consisting of a single vertex) and every
		path $P_{i,j}$, $0 \leq j \leq k$ does not contain any edge nor internal vertex on any path $Q_j$
		(we explicitly allow $P_{i,0}$ and $P_{i,k}$ to be paths of length $0$);
		\item a symmetric condition holds with the roles of $\mathcal{P}$ and $\mathcal{Q}$ swapped.
	\end{itemize}
\end{definition}

In what follows, the paths $P_{i,j}$, $P_{i,j}^Q$, $Q_{i,j}$, and $Q_{i,j}^P$ are considered empty for pairs of indices $(i,j)$ not defined above.

Clearly, if $G$ contains a $k \times k$-grid as a minor, it contains a $k$-pseudogrid:
just let the paths $\mathcal{P}$ follow the rows of the grid and the paths of $\mathcal{Q}$ follow the columns.
To finish the proof of Theorem~\ref{thm:grid_bound}, it suffices to show the following technical
result.
\begin{lemma}\label{lem:grid}
If $G$ contains a $k$-pseudogrid, then $\chilin(G) = \Omega(\sqrt{k})$.
\end{lemma}
\begin{proof}
Fix a linear coloring $\psi$ of $G$.
Let $(\mathcal{P}, \mathcal{Q})$ be a $k$-pseudogrid in $G$.
Let $V(\mathcal{P}) = \bigcup_{P \in \mathcal{P}} V(P)$ and similarly define $V(\mathcal{Q})$.
Let $\mu(\mathcal{P})$ be the number of distinct colors $\psi$ uses on $V(\mathcal{P})$ and similarly define $\mu(\mathcal{Q})$.
To prove the lemma, it suffices to show for any $k$-pseudogrid $(\mathcal{P},\mathcal{Q})$ in $G$,
   $k \leq 100 \cdot (\mu(\mathcal{P}) + \mu(\mathcal{Q}))^2$.
We shall prove it by induction over $k$.

The statement is trivial for $k \leq 100$.
For an inductive step, we proceed as follows.
For a vertex $v \in V(\mathcal{P}) \cup V(\mathcal{Q})$, the \emph{grid coordinate} of $v$ is $(i,j)$ if $v \in V(P_{i,j}^Q) \cup (V(P_{i,j}) \setminus V(P_{i,j+1}^Q)) \cup V(Q_{i,j}^P) \cup (V(Q_{i,j}) \setminus V(Q_{i+1,j}^P))$.
A vertex $v$ is \emph{marginal} if its grid coordinates $(i,j)$ satisfy $i \leq 3$, $j \leq 3$, $i \geq k-2$, or $j \geq k-2$.
A color $c$ is \emph{infrequent on $\mathcal{P}$} if it appears on $V(\mathcal{P})$, but there exists a family $\mathcal{P}_c \subseteq \mathcal{P}$
of at size at most $50(\mu(\mathcal{P}) + \mu(\mathcal{Q}))$ such that every vertex $v \in V(\mathcal{P})$ with $\psi(v) = c$ is either marginal
or lies on one of the paths in $\mathcal{P}_c$.
The definition of a color infrequent on $\mathcal{Q}$ is analogous.

For an inductive step, it suffices to show that there is always an infrequent color on $\mathcal{P}$ or an infrequent color on $\mathcal{Q}$.
Indeed, assume that $c$ is infrequent on $\mathcal{P}$ (the arguments for $\mathcal{Q}$ are symmetrical) and let $\mathcal{P}_c \subseteq \mathcal{P}$ be as in the above definition.
Construct a $k'$-pseudogrid $(\mathcal{P}',\mathcal{Q}')$ from $(\mathcal{P},\mathcal{Q})$ as follows.
Start with $(\mathcal{P}',\mathcal{Q}') = (\mathcal{P},\mathcal{Q})$. First, delete from $\mathcal{P}'$ the first and last $3$ paths, and similarly
for $\mathcal{Q}'$. Second, shorten every path $P_i \in \mathcal{P}'$ by deleting the edges of $P_{i,j}$ and $P_{i,j}^Q$ for $j \leq 3$ and $j \geq k-2$;
similarly shorten every path $Q_j \in \mathcal{Q}'$. Finally, delete all (shortened) paths of $\mathcal{P}_c$ from $\mathcal{P'}$, and delete a matching number
of paths from $\mathcal{Q}'$.
In this manner, we obtain a $k'$-pseudogrid $(\mathcal{P}',\mathcal{Q}')$ such that
    $k-k' \leq 6 + 50(\mu(\mathcal{P}) + \mu(\mathcal{Q}))$ and such that the color $c$ no longer appears on $V(\mathcal{P})$.
Therefore, $\mu(\mathcal{P}') + 1 \leq \mu(\mathcal{P})$ and $\mu(\mathcal{Q}') \leq \mu(\mathcal{Q})$.
The inductive step follows.

In the remainder of the proof, assume that there is no infrequent color on $\mathcal{P}$ nor an infrequent color on $\mathcal{Q}$.
We shall reach a contradiction
by exhibiting a simple noncentered path $P\subseteq \mathcal{P} \cup \mathcal{Q}$.

We perform the following selecting and marking scheme. Initially, no vertex is selected and no path is marked.
For every color $c$ that appears on $V(\mathcal{P})$, perform the following operation \emph{twice}.
\begin{enumerate}
	\item \label{i:grid:pick} Pick a vertex $v \in V(\mathcal{P})$ such that $\psi(v) = c$, $v$ is not marginal,
	  and $v$ does not lie on a marked path $P_i$. Let the grid coordinates of $v$ be $(i,j)$.
	\item \label{i:grid:mark} Select $v$ and mark all paths $P_{i'}$ for $|i'-i| \leq 10$ and all paths $Q_{j'}$ for $|j'-j| \leq 10$.
\end{enumerate}
Now swap the roles of $\mathcal P$ and $\mathcal Q$ and perform the above operation twice also for every color $c$ that appears on $V(\mathcal{Q})$.
In total, we select $2(\mu(\mathcal{P}) + \mu(\mathcal{Q}))$ vertices. For every selected vertex we mark $21$ paths of $\mathcal{P}$ and $21$ paths of $\mathcal{Q}$.
Since there is no infrequent color, there is always a vertex to choose at Step~\ref{i:grid:pick}, as otherwise the so-far marked paths would witness infrequency of $c$.
Thus, the above selecting and marking scheme is well-defined.

Let $v,v'$ be two distinct selected vertices and let $(i,j)$ and $(i',j')$ be their grid coordinates. By the above marking scheme, we have that
\begin{equation}\label{eq:grid:far}
3 < i,i',j,j' < k-2 \quad \mathrm{and} \quad |i-i'| + |j-j'| \geq 11.
\end{equation}

Consider now the following simple path $P$. We start with $P$ being the concatenation of even-numbered paths $P_i$ without the prefixes and suffixes $P_{i,0} \cup P_{i,k}$ in the natural order,
connected by paths $Q_{i,1} \cup Q_{i+1,1}^P \cup Q_{i+1,1}$ for $i$ divisible by $4$ and by $Q_{i,k} \cup Q_{i+1,k}^P \cup Q_{i+1,k}$ for $i\equiv 2 \pmod 4$
(so that paths $P_i$ with $i\equiv 2 \pmod 4$ are traversed forwards and paths $P_i$ with $i$ divisible by $4$ are traversed backwards).
Then, for every selected $v$ with grid coordinates $(i,j)$, we pick an even $i' \in \{i,i+1\}$ and modify locally $P \cap P_{i'}$ to pass through $v$.
In the modification, we use only parts of paths $P_{i,j}^Q \cup P_{i,j} \cup P_{i,j+1}^Q$, $P_{i+1,j}^Q \cup P_{i+1,j} \cup P_{i+1,j+1}^Q$,
$Q_{i,j}^P \cup Q_{i,j} \cup Q_{i+1,j}^P$, and $Q_{i,j+1}^P \cup Q_{i,j+1} \cup Q_{i+1,j+1}^P$.
By Equation~\eqref{eq:grid:far}, two such modifications do not interfere with each other and no such modification interferes with the connections contained in paths $Q_1$ and $Q_k$.
Consequently, the final path $P$ is a simple path contained in $\mathcal{P} \cup \mathcal{Q}$ that visits all selected vertices.
Such a path does not contain a center, which is the desired contradiction.
\qed\end{proof}

	\section{Treedepth Upper Bounds on Trees}\label{sec:trees}
\Schaffer proved that there is a linear time algorithm for finding a minimum-sized centered coloring of a tree $T$~\cite{schaeffer1989optimal}.
In this section we prove the following theorem by showing a correspondence between the centered coloring from \Schaffers algorithm and colors on paths in any linear coloring of $T$.
\begin{theorem}\label{thm:tree_bound}
Let $T$ be a tree of maximum degree $\Delta \geq 3$,
Then \Schaffers algorithm finds a centered coloring of $T$
with size at most $(\log_2 \Delta) \cdot \chilin(T)$.
%	There exists a polynomial time algorithm that takes as input a tree $T$ of maximum degree $\Delta \geq 3$
%  and a linear coloring $\psi$ of $T$ with size $k$ and outputs a centered coloring of $T$ whose size is at most $(\log_2 \min(\Delta, k))\cdot k$.
\end{theorem}
In particular, for trees of maximum degree $3$ we have $\chicen(T) \leq \log_2(3) \chilin(T)$,
   matching the lower bound of Lemma~\ref{lem:binary_tree_lc}.
We do not have any matching lower bound for larger $\Delta$.
In fact, we conjecture that none exists,
   that is, the upper bound of Theorem~\ref{thm:tree_bound} for $\Delta \geq 4$ is not tight.

\Schaffers algorithm finds a particular centered coloring whose colors are ordered in a way that reflects their roles as centers.
For this reason, the coloring is called a \emph{vertex ranking} and the colors are referred to as \emph{ranks}; it guarantees that in each subgraph, the vertex of maximum rank is also a center.
We will use this terminology in this section to clearly distinguish between the ranks in the vertex ranking and colors in the linear coloring.
Note that the canonical centered coloring of a treedepth decomposition is a vertex ranking if the colors are ranked decreasing from the root downwards, which implies that every centered coloring can be converted to a vertex ranking of the same size.
Of central importance to \Schaffers algorithm are what we will refer to as \emph{rank lists}.\looseness-1
\begin{definition}
	For a vertex ranking $r$ of tree $T$, the \emph{rank list} of $T$, denoted $L(T)$, can be defined recursively as $L(T) = L(T\backslash T_v)\cup \{r(v)\}$ where $v$ is the vertex of maximum rank in $T$.
\end{definition}
\Schaffers algorithm arbitrarily roots $T$ and builds the ranking from the leaves to the root of $T$, computing the rank of each vertex from the rank lists of each of its children.
For brevity, we denote $L(v) = L(T_v)$ for every $v$ in $T$.
\begin{proposition}[\cite{schaeffer1989optimal}]\label{prop:rank_lists}
	Let $r$ be a vertex ranking of $T$ produced by \Schaffers algorithm and let $v\in T$ be a vertex with children $u_1,\dots, u_\ell$.
	If $x$ is the largest integer appearing on rank lists of at least two children of $v$ (or 0 if all such rank lists are pairwise disjoint) then $r(v)$ is the smallest integer satisfying $r(v)> x$ and $r(v)\notin \bigcup_{i=1}^{\ell} L(u_i)$.
\end{proposition}
We root $T$ at an arbitrary leaf of $T$ and let $r$ be a ranking output by Schaffers algorithm
applied on (rooted) $T$.
With a vertex $v$ in $T$ we associate the following potential.
$$\zeta(v) = \sum_{r \in L(v)} 2^r.$$
The following is immediate from Proposition~\ref{prop:rank_lists}:
\begin{lemma}\label{lem:trees:rank}
For every $v$ in $T$ with children $u_1,u_2,\ldots,u_\ell$, it holds that
$$\zeta(v) \leq 2 + \sum_{i=1}^\ell \zeta(u_i).$$
Furthermore, the equality holds if and only if all rank lists $L(u_i)$ are pairwise disjoint.
\end{lemma}

Let $\psi$ be a linear coloring of $T$ with $k := \chilin(T)$ colors.
Our proof of Theorem~\ref{thm:tree_bound} is based on tracking sets of colors of $\psi$ on paths terminating at the current vertex as \Schaffers algorithm moves up the rooted tree.
Given a path $P\subseteq T$ and a linear coloring $\psi$ of size $k$, we say a \emph{color set} $X\subseteq \{1,\dots, k\}$ is \emph{compatible} with $P$ if both the following conditions are true:
\begin{enumerate}
	\item For every center $v\in P$, $\psi(v)\in X$.
	\item For every color $c\in X$, there is a vertex $u\in P$ such that $\psi(u) = c$.
\end{enumerate}
In other words, a compatible set must not contain colors not found on $P$, must contain each color appearing uniquely in $P$, and may or may not contain any colors appearing multiple times on $P$.
For each $v\in T$, let $S(v)$ be a set of sets defined recursively as follows.
If $v$ is a leaf, $S(v) = \{\{\psi(v)\}\}$.
Otherwise, let $u_1,\dots, u_\ell$ be the children of $v$, $S' = \bigcup_{i=1}^{\ell} S(u_i)$, and $\xi: S'\to 2^{[k]}$ be an injective function such that
$$
	\xi(X) =
	\begin{cases}
		X\setminus\{\psi(v)\} & \text{if } \psi(v)\in X \text{ and } X\backslash\{\psi(v)\} \in S' \\
		X\cup\{\psi(v)\} & \text{otherwise}
	\end{cases}
$$
for all $X\in S'$.
Then $S(v) = \{\xi(X): X\in S'\}\cup\{\{\psi(v)\}\}$.
We start with the following straightforward observation.
\begin{lemma}\label{lem:trees:bi}
For every $v \in T$ it holds that $\emptyset \notin S(v)$.
Consequently, for every nonleaf $v \in T$,
  $\xi$ is a bijection between $S'$ and $S(v) \setminus \{\{\psi(v)\}\}$.
\end{lemma}
We prove that the construction of $S(v)$ preserves compatibility of sets.
\begin{lemma}\label{lem:trees:paths}
	For all vertices $v\in T$ and each $X\in S(v)$, there is a corresponding path $P\subseteq T_v$ with $v$ as an endpoint such that $P$ is compatible with $X$.
\end{lemma}
\begin{proof}
	It is clear that the lemma holds at the leaves of $T$, so we proceed by inductively showing the recursive step preserves the property.
  Observe that the path consisting of $v$ only is compatible with $\{\psi(v)\} \in S(v)$.
  For any $X\in S(v) \setminus \{\{\psi(v)\}\}$, there is a child $u$ of $v$ such that $X' = \xi^{-1}(X)$ is in $S(u)$.
	By the inductive hypothesis, there must be a path $P'$ terminating at $u$ such that $X'$ is compatible with $P'$.
	We claim that $P = P'\cdot \{v\}$ is compatible with $X$.
	Since $X\triangle X'\subseteq \{\psi(v)\}$ and each color $c\neq \psi(v)$ appears the same number of times in $P$ and $P'$, it is only necessary to prove the requirements for compatibility are satified with respect to $\psi(v)$.
	Moreover, because $\psi(v)$ appears at least once on $P$ it suffices to show that if $\psi(v)\notin X$, then $\psi(v)$ appears multiple times on $P$.
	By the definition of $\xi$, $\psi(v)\notin X$ implies $\psi(v) \in X'$ and thus $v$ is not a center of $P$.
\qed\end{proof}

Define $\rho(v) = \sum_{X\in S(v)} (\Delta-1)^{|X|}$.
We observe the following
\begin{lemma}\label{lem:trees:rho}
	For any vertex $v \in T$ with children $u_1,\dots, u_\ell$,
  $\rho(v)\geq (\Delta-1) + \sum_{i=1}^{\ell} \rho(u_i)$.
\end{lemma}
\begin{proof}
First, note that $\ell \leq \Delta-1$.
Also, the lemma is straightforward for a leaf $v$ as then $S(v) = \{\{\psi(v)\}\}$ and
$\rho(v) = \Delta-1$. Assume then $\ell \geq 1$.

Recall that $S' = \bigcup_{i=1}^\ell S(u_i)$.
Let $S_1$ be the set of all color sets that appear in exactly one $S(u_i)$ and $S_M$ be those that occur in multiple $S(u_i)$'s; we have $S' = S_1 \uplus S_M$.
Note that for each $X\in S_M$, $\psi(v)\notin X$ or else concatenating the corresponding compatible paths with $v$ creates a path with no center.
Likewise, if there are distinct color sets $Y$ and $Y' = Y\backslash\{\psi(v)\}$ such that $\{Y,Y'\} \subseteq S_1\cup S_M$, then $Y,Y'$ both belong to the same $S(u_i)$;
in particular, both $Y,Y'$ belong to $S_1$.

By the definition of $\xi$, for each color set $X\in S(v) \setminus \{\{\psi(v)\}\}$
either $|X| \geq |\xi^{-1}(X)|$ or $|X|  = |\xi^{-1}(X)| - 1$.
In the latter case, there is a corresponding color set $X' = X\cup \{\psi(v)\}$ such that $X'\in S(v)$ and $\xi^{-1}(X') = X$. Also, from the discussion in the previous paragraph we infer
that this latter case can only happen when $X,X' \in S_1$.
Hence,
\begin{align*}
\sum_{X' \in S_1} (\Delta-1)^{|\xi(X')|} &\geq \sum_{X' \in S_1} (\Delta-1)^{|X'|},\ \mathrm{and} \\
\sum_{X' \in S_M} (\Delta-1)^{|\xi(X')|} &= (\Delta-1) \sum_{X' \in S_M} (\Delta-1)^{|X'|}.
\end{align*}

We infer that
\begin{align*}
\rho(v) &= (\Delta-1)^{|\{\psi(v)\}|} + \sum_{X \in S(v) \setminus \{\{\psi(v)\}\}} (\Delta-1)^{|X|} \\
        &= (\Delta-1) + \sum_{X' \in S'} (\Delta-1)^{|\xi(X')|} \\
        &= (\Delta-1) + \sum_{X' \in S_1} (\Delta-1)^{|\xi(X')|} + \sum_{X' \in S_M} (\Delta-1)^{|\xi(X')|} \\
        &\geq (\Delta-1) + \sum_{X' \in S_1} (\Delta-1)^{|X'|} + (\Delta-1)\sum_{X' \in S_M} (\Delta-1)^{|X'|} \\
        &\geq (\Delta-1) + \sum_{i=1}^\ell \sum_{X' \in S(u_i)} (\Delta-1)^{|X'|} \\
        &\geq (\Delta-1) + \sum_{i=1}^\ell \rho(u_i).
\end{align*}
\qed\end{proof}
We conclude with the proof of Theorem~\ref{thm:tree_bound}.

\begin{proof}[Theorem~\ref{thm:tree_bound}]
For every leaf $v \in T$, we have $\rho(v) = \Delta-1 \geq 2 = \zeta(v)$.
Lemmas~\ref{lem:trees:rank} and~\ref{lem:trees:rho} show inductively
that $\rho(v) \geq \zeta(v)$ for every $v \in T$.
If $k'$ is the size of the centered coloring output by \Schaffers algorithm, then
for the root $v_0$ of $T$ we have
$$2^{k'} \leq \zeta(v_0) \leq \rho(v) \leq \sum_{X \subseteq [k]} (\Delta-1)^{|X|} = \Delta^k.$$
Thus $k' \leq (\log_2 \Delta) \cdot k$.
\qed\end{proof}

	\section{Treedepth Upper Bounds on Interval Graphs}\label{sec:interval}
	Because linear colorings are equivalent to centered colorings when restricted to paths, we turn our attention to the linear coloring numbers of ``pathlike'' graphs.
	We investigate a particular class of ``pathlike'' graphs in this section and prove a quadratic relationship between their centered and linear coloring numbers.
	\begin{definition}\label{def:interval}
		A graph $G$ is an \emph{interval graph} if there is an injective mapping $f$ from $V(G)$ to intervals on the real line such that $uv\in E(G)$ iff $f(u)$ and $f(v)$ overlap.
	\end{definition}
	We refer to the mapping $f$ as the \emph{interval representation} of $G$.
	Since the overlap between intervals $f(u)$ and $f(v)$ is independent of the interval representations of the other vertices, every subgraph of an interval graph is also an interval graph.
	The interval representation of $G$ implies a natural ``left-to-right'' layout that gives it the ``pathlike'' qualities, which are manifested in restrictions on the length of induced cycles (\emph{chordal}) and paths between vertex triples (\emph{AT-free}).
	\begin{definition}
		A graph is \emph{chordal} if it has no induced cycles of length $\geq 4$.
	\end{definition}
	\begin{definition}\label{def:at}
		Vertices $u,v,w$ are an \emph{asteroidal triple} (AT) if there exist $uv$-, $vw$-, and $wu$-paths $P_{uv}$, $P_{vw}$, and $P_{wu}$, respectively, such that $N[w]\cap P_{uv} = N[u]\cap P_{vw} = N[v]\cap P_{uv} = \emptyset$.
		A graph with no AT is called \emph{AT-free}.
	\end{definition}
	\begin{proposition}[\cite{lekkeikerker1962representation}]\label{prop:chordal_at-free}
		A graph $G$ is an interval graph iff $G$ is chordal and AT-free.
	\end{proposition}
	Intuitively, Definition~\ref{def:at} is a set of three vertices such that every pair is connected by a path that avoids the neighbors of the third.
	Roughly speaking, in the context of linear colorings, Proposition~\ref{prop:chordal_at-free} indicates that if $w$ is a center of a ``long'' $uv$-path $P$ in $G$, any vertex $w'$ such that $\psi(w) = \psi(w')$ must have a neighbor on $P$.
	We devote the rest of this section to proving Theorem~\ref{thm:interval_bound}.
	\begin{theorem}\label{thm:interval_bound}
		There exists a polynomial time algorithm that takes as input an interval graph $G$ and a linear coloring of $G$ with size $k$ and outputs a centered coloring of $G$ with size at most $k^2$.
	\end{theorem}
	Our algorithm makes extensive use of the following well-known property of maximal cliques in interval graphs.
	\begin{proposition}[\cite{lekkeikerker1962representation}]\label{prop:clique_order}
		If $G$ is an interval graph, its maximal cliques can be linearly ordered in polynomial time such that for each vertex $v$, the cliques containing $v$ appear consecutively.
	\end{proposition}
	In particular, we identify a \emph{prevailing path} in $G$ whose vertices ``span'' the maximal cliques and a \emph{prevailing subgraph} that consists of the prevailing path as well as vertices in maximal cliques ``between'' consecutive vertices on the prevailing path.
	We will show that any linear coloring is a centered coloring when restricted to the prevailing subgraph and that after removing the prevailing subgraph, the remaining components each use fewer colors.

	Let $C_1, \dots C_m$ be an ordering of the maximal cliques of $G$ that satisfies Proposition~\ref{prop:clique_order}.
	We say vertex $v$ is \emph{introduced} in $C_{i}$ if $v\in C_i$ but $v\notin C_{i-1}$, and denote this as $I(v) = i$.
	Likewise, $v$ is \emph{forgotten} in $C_j$ if $v\in C_{j}$ but $v\notin C_{j+1}$, and denote this as $F(v) = j$.
	The procedure for constructing a prevailing subgraph and prevailing path is described in Algorithm~\ref{alg:rho}.
	This algorithm selects the vertex $v$ from the current maximal clique that is forgotten ``last'' and adds $v$ to the prevailing path and $C_{F(v)}$ to the prevailing subgraph.
	We prove in Lemma~\ref{lem:rho_hamiltonian} that if $P,Q$ are a prevailing path and subgraph, the vertices in $Q\backslash P$ can be inserted between vertices of $P$ to form a Hamiltonian path of $Q$.\looseness-1
	\begin{algorithm}
\begin{algorithmic}[1]
	\REQUIRE interval graph $G$
	\ENSURE prevailing path $P$ and prevailing subgraph $Q$
	\STATE $C_1,\dots, C_m\leftarrow$ maximal cliques of $G$ labeled in accordance with Proposition~\ref{prop:clique_order}
	\STATE $P\leftarrow \emptyset$
	\STATE $V_Q\leftarrow \emptyset$
	\STATE $i\leftarrow 1$
	\STATE $j\leftarrow 1$
	\WHILE{$i<m$}
		\STATE $v_j\leftarrow \argmax_{u\in C_i} F(u)$\label{alg_step:select_v}
		\STATE $P\leftarrow P \cdot \{v_j\}$\label{alg_step:add_to_P}
		\STATE $i\leftarrow F(v)$
		\STATE $V_Q\leftarrow V_Q\cup V(C_i)$
		\STATE $j\leftarrow j+1$
	\ENDWHILE
	\STATE $Q\leftarrow G[V_Q]$
	\RETURN $P, Q$
\end{algorithmic}
\caption{Construction of a prevailing path and subgraph.}\label{alg:rho}
\end{algorithm}

	\begin{lemma}\label{lem:rho_hamiltonian}
		Every prevailing subgraph has a Hamiltonian path.
	\end{lemma}
	\begin{proof}
	Let $P,Q$ be the prevailing path and subgraph constructed in Algorithm~\ref{alg:rho}.
	We prove by constructing the Hamiltonian path of $Q$.
		% By this definition, $v_{j+1}\in C_{F(v_j)}$ and thus there is a path $P=\{v_1,\dots, v_p\}$.
		Let $M_j$ be the set of all $u\in Q\backslash P$, for which $j$ is the smallest integer for which $u\in C_{F(v_j)}$.
		In other words $M_j$ contains the vertices in $C_{F(v_j)}$ that do not appear in $C_{F(v_{j-1})}$.
		If $\mathcal M =  \bigcup_{1\leq j\leq p} M_j$ then by construction $P\cup \mathcal M = Q$.
		Moreover, for each $u\in M_j$, $M_j\cup \{v_j,v_{j+1}\} \subseteq N[u]$.
		For each $M_j$, let $\mu^{1}_j,\mu^2_j,\dots, \mu^{|M_j|}_j$ be a ordering of $M_j$ such that $F(\mu^i_j) \leq F(\mu^{i+1}_j)$.
		Then
			$$v_1,\mu^1_1,\dots, \mu^{|M_1|}_1, v_2,\mu^{1}_2,\dots,\mu^{|M_2|}_2, \dots, v_p,\mu^{1}_p,\dots \mu^{|M_p|}_p$$
		 is a Hamiltonian path.
	\hfill\qed\end{proof}

	Although the fact that the prevailing subgraph $Q$ has a Hamiltonian path implies $Q$ has a center with respect to $\psi$, we must ensure that the proper subgraphs of $Q$ also have a center.
	In Lemma~\ref{lem:rho_centered}, we prove $\psi|_Q$ is centered by showing every proper connected subgraph of $Q$ also has a Hamiltonian path.

	\begin{lemma}\label{lem:rho_centered}
		If $Q$ is a prevailing subgraph of an interval graph $G$ and $\psi$ a linear coloring of $G$, $\psi|_{Q}$ is a centered coloring.
	\end{lemma}
	\begin{proof}
		It suffices to show that every proper, connected induced subgraph of $Q$ has a Hamiltonian path, since the existence of a Hamiltonian path implies the subgraph has a center.
		Assume $H\subseteq Q$ has a Hamiltonian path.
		Let $w$ be a center and $w_p, w_s$ be its predecessor and successor in the Hamiltonian path.
		It is clear that the subpath from the start of the Hamiltonian path to $w_p$ remains a path in $H\backslash \{w\}$; this is also true for the subpath from $w_s$ to the end.
		Therefore if $H\backslash\{w\}$ is disconnected, there are two components and both have Hamiltonian paths.

			Otherwise suppose $H\backslash \{w\}$ is connected.
		Note that if $P=\{v_1,\dots, v_p\}$ is the prevailing path generated by Algorithm~\ref{alg:rho}, $C_{F(v_j)}\cap C_{F(v_{j+2})} = \emptyset$ or else $v_{j+2}$ would be forgotten later than $v_{j+1}$ and would have been chosen to be $v_{j+1}$ instead.
		Thus, there is some $1\leq \ell\leq q\leq p$ such that $H\backslash\{w\}\subseteq C_{F(v_\ell)}\cup C_{F(v_{\ell+1})}\cup \dots \cup C_{F(v_q)}$ and since $H\backslash\{w\}$ is connected, for each $\ell\leq j< q$ the intersection of cliques $C_{F(v_j)}$ and $C_{F(v_{j+1})}$ is non-empty.
		Consequently, the ordering of the vertices in the Hamiltonian path of $Q$ must also define a Hamiltonian path of $H\backslash \{w\}$.
	\hfill\qed\end{proof}
	Since any linear coloring $\psi$ of the prevailing subgraph $Q$ must also be a centered coloring, $\td(Q)\leq |\psi|$.
	To get a bound on the treedepth of $G$, we focus on the relationship between $Q$ and $G\backslash Q$.
	In particular, we show that the components of $G\backslash Q$ use fewer than $|\psi|$ colors by proving that each such component has an apex in the prevailing path.
	\begin{lemma}\label{lem:rho_apex}
		Let $P,Q$ be a prevailing path and subgraph of an interval graph $G$.
		For each component $X$ of $G\backslash Q$, there is a vertex $a\in P$ such that $X\subseteq N(a)$.
	\end{lemma}
	\begin{proof}
		For $1\leq j \leq p$, let $\mathcal{X}_j$ be the set of components of $G[\bigcup_{i = F(v_{j-1})+1}^{F(v_j)-1} C_i]\backslash Q$, defining $F(v_0) = 0$.
		By this definition and the fact that $v_j$ is a member of both $C_{F(v_{j-1})}$ and $C_{F(v_j)}$,  $v_j$ is a neighbor of all vertices in $X$ for each $X\in \mathcal{X}_j$.
		Thus it suffices to show that $\bigcup_{j=1}^{p} \mathcal{X}_j$ are the components of $G\backslash Q$.

			Since $V(Q) = \bigcup_{j=1}^{p} C_{F(v_j)}$, $V(G) = V(Q) \cup V(\mathcal{X}_1) \cup \dots\cup V(\mathcal X_p)$ and $V(Q)\cap \bigcup_{j=1}^{p} \mathcal X_j = \emptyset$.
		Hence, if $X\in\mathcal{X}_j$ is not a component of $G\backslash Q$, then there must be some component $X'\in \mathcal{X}_i$ for which $i\neq j$ and there exists $u\in X$ and $u'\in X'$ and $uu'\in E(G)$.
		But $C_{F(v_j)}\cup C_{F(v_{j+1})}$ has no common vertices with $X$ and separates it from any vertices in $\mathcal{X}_i$.
		An analogous statement for $X'$ is true as well, so no such edges $uu'$ exist.
		Therefore we conclude that $\bigcup_{1\leq j\leq p} \mathcal{X}_j$ are the components of $G\backslash Q$ and the lemma is proven.
	\hfill\qed\end{proof}

	We can now establish a polynomial upper bound on the treedepth of interval graphs, proving Theorem~\ref{thm:interval_bound}.
	% proof of main theorem
	\begin{proof}[Theorem~\ref{thm:interval_bound}]
		Let $\mathcal A$ be the algorithm that constructs a treedepth decomposition $\mathcal T$ of $G$ by finding a prevailing subgraph $Q$ (Algorithm~\ref{alg:rho}), using $\psi|_{Q}$ to create a treedepth decomposition of $Q$, and recursively constructing treedepth decompositions of $G\backslash Q$.
		If $\operatorname{depth}(\mathcal T) \leq k^2$ and $\mathcal A$ runs in polynomial time, then the canonical centered coloring of $\mathcal T$ is a centered coloring of $G$ of size at most $k^2$.
		We prove $\mathcal A$ satisfies these requirements by induction on $k = |\psi|$.
		At $k = 1$, the graph consists of isolated vertices and $\mathcal{A}$ trivially constructs a treedepth decomposition of $G$ of depth $1$ in polynomial time.

		Assume $\mathcal A$ has the desired properties for linear colorings of size at most $k-1$.
		Because the maximal cliques of an interval graph can be enumerated and ordered in polynomial time (Proposition~\ref{prop:clique_order}), identifying $Q$ via Algorithm~\ref{alg:rho} can be done in polynomial time.
		By Lemma~\ref{lem:rho_centered}, the canonical treedepth decomposition of $Q$ has depth at most $k$.
		Since every component $X$ of $G\backslash Q$ has an apex $a$ in $P$ (Lemma~\ref{lem:rho_apex}), we can assume $a$ is an ancestor in $\mathcal T$ of each vertex in $X$ (Lemma~\ref{lem:apex_order}).
		Because $\psi$ is proper, $\psi(a)$ does not appear in $\psi|_X$ and since induced subgraphs of interval graphs are themselves interval graphs, $\mathcal A$ finds a treedepth decomposition of $X$ whose depth is at most $(k-1)^2$.
		Thus $\mathcal T$ has depth $k+(k-1)^2\leq k^2$.
		The recursion only lasts $k\leq n$ steps, so $\mathcal A$ runs in polynomial time.
	\hfill\qed\end{proof}

	\section{Hardness of Recognizing Linear Colorings}
Based on the similarity in definition between linear and centered colorings, one might assume that computing them should be roughly equally difficult.
Finding a centered coloring of a fixed size is NP-hard~\cite{bodlaender1998rankings}, but given a coloring of a graph, we can recognize whether it is centered in polynomial time by attempting to create the canonical treedepth decomposition; this procedure will identify a non-centered subgraph if the coloring is not centered.
To the contrary, we will prove that \textsc{Linear Coloring Recognition}, the problem of \emph{recognizing} whether a coloring is linear, is co-NP-complete.
In order to prove the hardness of \textsc{Linear Coloring Recognition}, we first define a dual problem.
The \textsc{Non-centered Path} problem takes a graph $G$ and coloring $\psi$ as input and decides whether $G$ has a non-centered path $P$.
We focus on proving the hardness of \textsc{Non-centered Path} because a certificate to that problem is easily definable:  a path where every color appears at least twice.

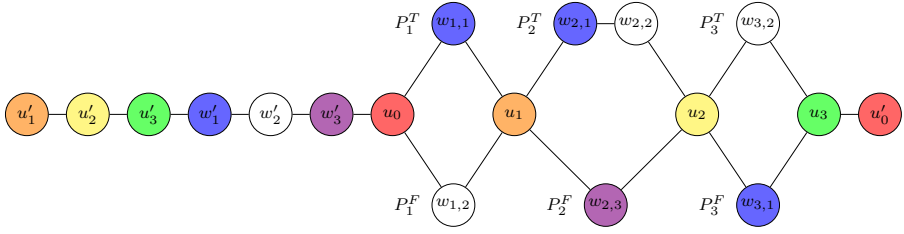
\begin{figure}
	\centering
	\resizebox{\linewidth}{!}{\begin{tikzpicture}
	\node[vertex,fill=orange!60] (u'1) at (-5,0) {$u'_1$};
	\node[vertex,fill=yellow!60] (u'2) at (-4,0) {$u'_2$};
	\node[vertex,fill=green!60 ] (u'3) at (-3,0) {$u'_3$};
	\node[vertex,fill=blue!60  ] (w'1) at (-2,0) {$w'_1$};
	\node[vertex,fill=white!60 ] (w'2) at (-1,0) {$w'_2$};
	\node[vertex,fill=violet!60 ] (w'3) at (0,0) {$w'_3$};

	\node[vertex,fill=red!60   ] (u0) at (1,0) {$u_0$};
	\node[vertex,fill=orange!60] (u1) at (3,0) {$u_1$};

	\node[vertex,fill=yellow!60] (u2) at (6,0) {$u_2$};

	\node[vertex,fill=green!60 ] (u3) at (8,0) {$u_3$};

	\node[vertex,fill=red!60   ] (u'0) at (9,0) {$u'_0$};

	\node[vertex,fill=blue!60  ] (w11) at (2,1.5) {$w_{1,1}$};
	\node[vertex,fill=white!60 ] (w12) at (2,-1.5) {$w_{1,2}$};
	\node[vertex,fill=blue!60  ] (w21) at (4,1.5) {$w_{2,1}$};
	\node[vertex,fill=white!60 ] (w22) at (5,1.5) {$w_{2,2}$};
	\node[vertex,fill=violet!60 ] (w23) at (4.5,-1.5) {$w_{2,3}$};
	\node[vertex,fill=blue!60  ] (w31) at (7,-1.5) {$w_{3,1}$};
	\node[vertex,fill=white!60 ] (w32) at (7,1.5) {$w_{3,2}$};

	%%%%%%%%%%%%%%%%%%%%%%%%
	%     L A B E L S
	%%%%%%%%%%%%%%%%%%%%%%%%
	\node (T1) at (1.25,1.5) {$P_1^T$};
	\node (T2) at (3.25,1.5) {$P_2^T$};
	\node (T3) at (6.25,1.5) {$P_3^T$};
	\node (F1) at (1.25,-1.5) {$P_1^F$};
	\node (F2) at (3.75,-1.5) {$P_2^F$};
	\node (F3) at (6.25,-1.5) {$P_3^F$};

	\draw (u0) -- (w11) -- (u1);
	\draw (u0) -- (w12) -- (u1);
	\draw (u1) -- (w21) -- (w22) -- (u2);
	\draw (u1) -- (w23) -- (u2);
	\draw (u2) -- (w31) -- (u3);
	\draw (u2) -- (w32) -- (u3);

	\draw (u'1) -- (u'2) -- (u'3) -- (w'1) -- (w'2) -- (w'3) -- (u0);
	\draw (u3) -- (u'0);

\end{tikzpicture}}
	\caption{The graph $G$ and coloring $\psi$ for $\Phi=(x_1\vee x_2 \vee \neg x_3)\wedge (\neg x_{1}\vee x_2 \vee x_3)\wedge (\neg x_2)$.}\label{fig:co_np_hardness}
\end{figure}

\begin{theorem}\label{thm:non-centered}
	\textsc{Non-centered Path} is NP-complete.
\end{theorem}
\begin{proof}
	A certificate to \textsc{Non-centered Path} can be verified in linear time by iterating over all vertices in the path and counting color occurrences.
	Thus, \textsc{Non-centered Path} is in NP.

		We prove NP-hardness by reducing from \textsc{CNF-SAT}.
	Given a \textsc{CNF-SAT} formula $\Phi$ with variables $x_1,\dots x_n$ and clauses $C_1,\dots C_m$, we construct a graph $G$ and coloring $\psi$ that will have a non-centered path if and only if $\Phi$ is satisfiable.
	We assume that $\Phi$ satisfies the following properties:
	\begin{enumerate}[label=(\arabic*)]
		\item Every variable appears at most once in each clause.
		\item No clause contains both a variable and its negation.
		\item Every variable appears as a positive literal and negative literal.
	\end{enumerate}
	We can assume (1) since the disjunction operation is idempotent.
	Every clause for which (2) does not hold is satisfied by any truth assignment of the variables and thus can be removed without changing the satisfiability of $\Phi$.
	If variable $x_i$ appears only positively then assigning $x_i$ to be false does not cause any clauses to be satisfied.
	Therefore, it is sufficient to set $x_i$ to true and only consider the clauses of $\Phi$ that do not contain $x_i$; since the analogous statement is true when $x_i$ does not appear positively, we can assume (3).

		The variables of $\Phi$ are represented by a set of vertices $U = \{u_0,\dots, u_n\}$.
	For each $x_i$, we connect $u_{i-1}$ and $u_i$ with two paths $P^{T}_i$ and $P^F_i$; we will force the non-centered path to contain vertices from exactly one of $P^T_i$ and $P^F_i$, which will correspond to whether $x_i$ was set to true or false.
	The path $P^T_i$ contains one vertex for each $C_j$ in which $x_i$ appears positively while $P^F_i$ contains one vertex for each $C_j$ in which the negation of $x_i$ appears.
	By assumption (2), we can uniquely label the vertex on $P^T_i\cup P^F_i$ corresponding to clause $C_j$ as $w_{i,j}$ and the order of the vertices on $P^T_i$ and $P^F_i$ can be chosen arbitrarily.
	To complete the construction of $G$, we add path $P_{0} = u'_1, u'_2, \dots, u'_n, w'_1, w'_2, \dots, w'_m$ such that $w'_m$ is adjacent to $u_0$ and all other vertices on $P_0$ have no additional edges.
	Finally, we attach a pendant vertex $u'_0$ to $u_n$.
	Since each vertex $w_{i,j}$ corresponds to a unique literal in $\Phi$ and $|U| + |P_0| = 2n+m+1$, $G$ has size linear in the size of $\Phi$.

		To encode satisfaction of clauses, we color $G$ with coloring $\psi: V(G)\to \{0,\dots, n+m\}$ such that $\psi(u_i) = \psi(u'_i) = i$ and $\psi(w'_j) = \psi(w_{i,j}) = n+j$.
	In this way, we force any non-centered path to contain all colors and color $j+n$ appears twice if and only if $C_j$ is satisfied.
	An example can be found in Figure~\ref{fig:co_np_hardness}.

		We now prove that $\Phi$ is satisfiable iff $G$ contains a path $Q$ with no center.
	Given a satisfying assignment of $\Phi$, let $P^*_i$ be $P^T_i$ if $x_i$ is set to true and $P^F_i$ if $x_i$ is set to false.
	Then $Q = P_0\cdot u_0 \cdot P^*_1\cdot u_1 \cdot \dots \cdot P^*_n \cdot u_n \cdot u'_0$, is a non-centered path since it contains all pairs $u_i, u'_i$ and $\bigcup_{1\leq i\leq n} P^*_i$ contains a vertex with the same color as each vertex in $\bigcup_{1\leq j\leq m} w'_j$.

		To prove the reverse direction suppose $G$ contains a non-centered path $Q$.
	Let $U' = \{u'_0,\dots, u'_n\}$.
	Since each vertex in $U'$ shares a color with exactly one other vertex and that vertex is a member of $U$, $Q$ contains a vertex from $U$ iff $Q$ contains a vertex from $U'$.
	By our construction of $P_0$ and assumptions about $\Phi$, no component of $G\backslash (U\cup U')$ contains two vertices of the same color.
	Thus, $Q$ must contain vertices from $U$, $U'$, and $G\backslash (U\cup U')$.
	For any $0\leq i\neq j\leq n$, every $u_ju'_j$ path contains $u_i$ or $u'_i$, which implies that $(U\cup U')\subset Q$ and $Q$ is a $u'_1u'_0$ path.
	In order for $Q$ to be connected, $P_0\subseteq Q$ and in order for it to be a path, exactly one of $P^T_i$ and $P^F_i$ (denote it $P^*_i$) is a subpath of $Q$ for each $1\leq i\leq n$.
	Since the colors in $w'_1,\dots w'_m$ are unique, $\bigcup_{1\leq i\leq n} P^*_i$ contains at least one vertex of each color on $[n+1, n+m]$, which corresponds to a selection of truth assignments to the variables of $\Phi$ such that every clause is satisfied.
\hfill\qed\end{proof}

\begin{corollary}
	\textsc{Linear Coloring Recognition} is co-NP-complete.
\end{corollary}

The co-NP-hardness of recognizing linear colorings is compounded by three stronger hardness implications.
First, the coloring $\psi$ given in Theorem~\ref{thm:non-centered} has size $m+n+1$, which means that unless the exponential time hypothesis~\cite{impagliazzo2001complexity} fails, there is no $2^{o(k)}$ algorithm to recognize a linear coloring of size $k$.
Second, the graph $G$ constructed in the proof of Theorem~\ref{thm:non-centered} is outerplanar with pathwidth two, which implies that neither treewidth-style dynamic programming nor a Baker-style layering approach is likely to solve this problem efficiently.
Finally, by subdividing each edge and coloring all subdivision vertices with a (single) new color, we obtain a bipartite graph with degeneracy two, proving hardness for each of those classes.
Nonetheless, the fact that $\chicen(G) = O(\log m + \log n)$ while $|\psi| = m+n+1$ leaves open the possibility that \textsc{Linear Coloring Recognition} becomes easier for colorings of minimum size.

	\section{Conclusion}
We have introduced $p$-linear and linear colorings as an alternative to $p$-centered and centered colorings for use in algorithms for classes of bounded expansion.
The $p$-linear colorings are computable in polynomial time and require a constant number of colors in classes of bounded expansion, while inducing graphs of bounded treedepth for all small sets of colors, allowing direct substitution in existing algorithmic pipelines.
A major direction for future work is to bring the upper bound on $\tmax(k)$ of $\operatorname{poly}(k)$ closer to the lower bound of $2k$.
In particular, it appears our current toolkit for analyzing linear colorings must be expanded in order to prove (or disprove) Conjecture~\ref{conj:two_chilin}.
We also believe it is worth studying whether recognizing linear colorings can be done in polynomial time if we assume the coloring is of size $\chilin(G)$.
Finally, using $p$-linear colorings in practice will require an efficient method for translating a linear coloring into a treedepth decomposition.
Although there exist general-purpose algorithms to find treedepth decompositions efficiently in graphs of bounded linear coloring number (e.g.~\cite{treedepthExact}), a more specialized algorithm that avoids ``heavy machinery'' is likely necessary to be practically useful.

\section*{Acknowledgments}
\begin{wrapfigure}{r}{0.17\textwidth}
\includegraphics[width=.8\linewidth]{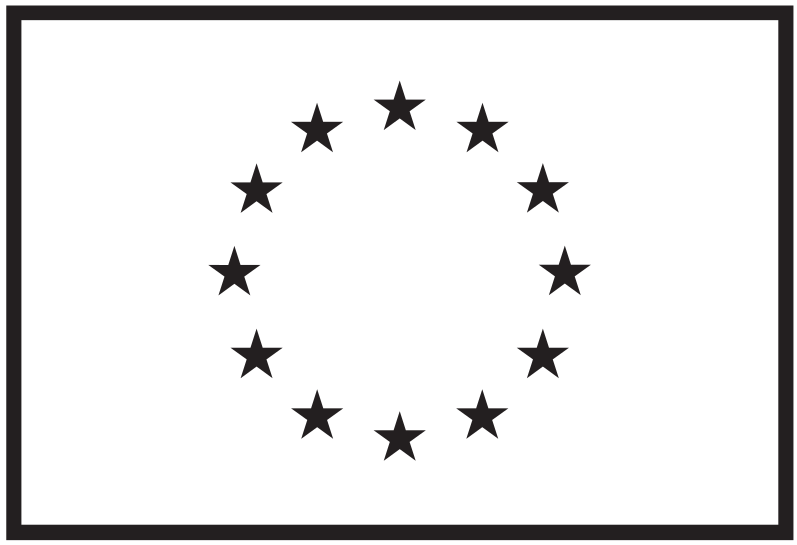}\\[3mm]
\includegraphics[width=.8\linewidth]{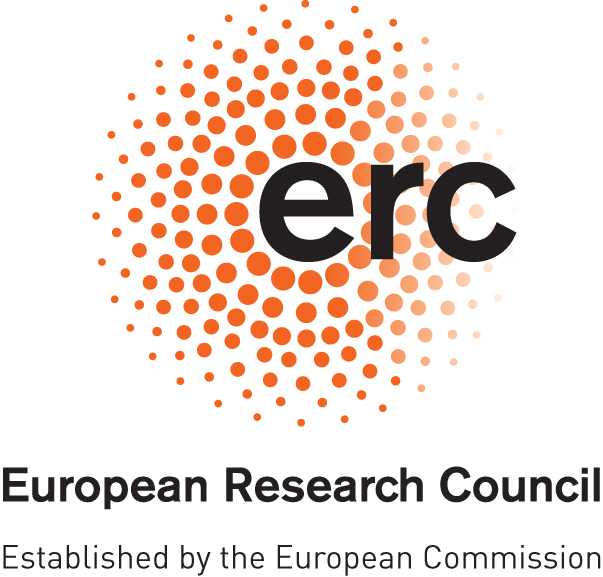}
\end{wrapfigure}
The authors would like to thank Felix Reidl and Fernando S\'{a}nchez-Villaamil for bringing these colorings to our attention and several anonymous reviewers for their helpful suggestion.
This work was supported in part by the DARPA GRAPHS Program and the Gordon \& Betty Moore Foundation's Data-Driven Discovery Initiative through Grants SPAWARN66001-14-1-4063 and GBMF4560 to Blair D. Sullivan.
The research of Marcin Pilipczuk is a part of a project that has received funding from the European Research Council (ERC) under the European Union's Horizon 2020 research and innovation programme, grant agreement No 714704.

	\bibliographystyle{spmpsci}
	\bibliography{lc.bib}
\end{document}